\UseRawInputEncoding
\documentclass{amsart}
\usepackage{amsmath,amsfonts,amscd,amssymb,amsthm,epsfig,euscript}
\usepackage{mathrsfs}
\usepackage{amsxtra}
\usepackage[all]{xy}
\usepackage{verbatim,epsf}
\usepackage{framed}
\newtheorem{theorem}{Theorem}
\newtheorem{proposition}{Proposition}
\newtheorem{lemma}{Lemma}
\newtheorem{corollary}{Corollary}
\newtheorem{question}{Question}
\theoremstyle{definition}
\newtheorem{definition}{Definition}

\theoremstyle{remark}
\newtheorem{remark}{Remark}
\numberwithin{equation}{section}
\newcommand{\field}[1]{\ensuremath{\mathbb{#1}}}

\newcommand{\RR}{\field{R}}
\newcommand{\TT}{\field{T}}

\newcommand{\ZZ}{\field{Z}}

\newcommand{\curly}[1]{\mathscr{#1}}
\newcommand{\cA}{\curly{A}}
\newcommand{\cB}{\curly{B}}

\newcommand{\cF}{\curly{F}}
\newcommand{\cG}{\curly{G}}

\newcommand{\cI}{\curly{I}}

\newcommand{\cL}{\curly{L}}

\newcommand{\cU}{\curly{U}}

\newcommand{\NN}{\field{N}}

 \DeclareMathOperator{\Hom}{Hom}

\newcommand{\noop}[1]{}

\usepackage{hyperref}
\hypersetup{colorlinks,citecolor=blue,plainpages=false,hypertexnames=false}
\begin{document}

\title[Thin homotopy and the holonomy approach]{Thin homotopy and the holonomy approach to gauge theories}


\author[Meneses]{Claudio Meneses}
\address{Mathematisches Seminar, Christian-Albrechts Universit\"at zu Kiel, Heinrich-Hecht- \indent Platz 6, 
24118 Kiel, Germany}
\email{meneses@math.uni-kiel.de} 

\thanks{}

\thanks{}


\date{}

\dedicatory{}
\begin{abstract}
We survey several mathematical developments in the holonomy approach to gauge theory. A cornerstone of this approach is the introduction of group structures on spaces of based loops on a smooth manifold, relying on certain homotopy equivalence relations --- such as the so-called thin homotopy --- and the resulting interpretation of gauge fields as group homomorphisms to a Lie group $G$ satisfying a suitable smoothness condition, encoding the holonomy of a gauge orbit of smooth connections on a principal $G$-bundle. We also prove several structural results on thin homotopy, and in particular we clarify the difference between thin equivalence and retrace equivalence for piecewise-smooth based loops on a smooth manifold, which are often used interchangeably in the physics literature. We conclude by listing a set of questions on topological and functional analytic aspects of groups of based loops, which we consider to be fundamental to establish a rigorous differential geometric foundation of the holonomy formulation of gauge theory. 
\\

\end{abstract}

\maketitle


\tableofcontents
\section{Introduction}
Since the early developments of gauge theory, the idea of ``nonintegrable phase factors" of a gauge field has been recognized as a fundamental notion on which the study of gauge symmetry could be built \cite{AB59,Man62a}. In the same way that gauge fields can be modeled mathematically as orbits in a space of smooth connections on a smooth principal $G$-bundle $P\rightarrow M$ under the action of the infinite dimensional group of bundle automorphisms \cite{WY75}, the phase factors associated to a gauge field can be modeled geometrically in terms of the holonomy (or more precisely, its conjugacy class in $G$) of any such an orbit representative around any given based loop \cite{Yang74}.  Although the mathematical idea of reconstructing a gauge orbit of smooth connections on a principal $G$-bundle in terms of the holonomies of orbit representatives is rather old and seems to have originated in the work of Kobayashi \cite{Koba54}, considering the latter as a substitute for the former is usually reserved for connections satisfying further integrability conditions such as flatness. An apparent reason that justifies this state of affairs is the following. While the holonomy of a flat connection depends only on the structure of the fundamental group $\pi_{1}(M,p)$ of the pointed base manifold $(M,p)$, the holonomy of an arbitrary smooth connection would depend on the structure of a considerably more complicated group $\cL^{\boldsymbol{\cdot}}(M,p)$ of equivalence classes of based loops, constructed in terms of weaker homotopy equivalences for which holonomy is a well-defined invariant. Thus, in the same way that a gauge orbit of flat connections may be equivalently characterized as a $G$-conjugacy class of group homomorphisms $\rho:\pi_{1}(M,p)\rightarrow G$, a gauge orbit of smooth connections can be equivalently characterized in terms of a $G$-conjugacy class of \emph{holonomy homomorphisms} \cite{Koba54,Bar91,Lew93,CP94}
\[
H:\cL^{\boldsymbol{\cdot}}(M,p)\rightarrow G
\]
satisfying certain suitable smoothness conditions. However, the topological groups $\cL^{\boldsymbol{\cdot}}(M,p)$ are still poorly understood, and do not seem to be widely known to differential geometers.  

It is thus not surprising that the study of gauge symmetry in terms of holonomy and parallel transport --- the so-called loop approach in the physics literature --- has not experienced a similar dramatic development into a mature mathematical discipline when compared to the more standard approaches to mathematical gauge theory. Any serious intent to rectify this unfortunate missed opportunity would necessarily rely on the development of pertinent analytic tools in the study of the geometry of spaces of based loops in a manifold, similar in spirit to the general programme elucidated in the works \cite{KM97,Sta09}.

Two loop homotopy relations that are often used interchangeably in the physics literature are those of $\emph{thin homotopy}$ \cite{Bar91,CP94} and \emph{retrace homotopy} \cite{Koba54,Tel60,Lew93}. Both relations lead to equivalence relations that allow to define groups of based loops $\cL^{\boldsymbol{\cdot}}(M,p)$. However, while retrace homotopy is a particular case of thin homotopy, it is possible to construct examples of thin homotopies that cannot be factored as a finite concatenation of reparametrizations and retracings. One of the main objectives of this article is to clarify the difference between these relations when defined over the space of piecewise-smooth based loops $\Omega^{\text{ps}}(M,p)$ of a smooth pointed manifold $(M,p)$, which is presented in detail in sections \ref{sec:thin} and \ref{sec:results}. On the other hand, the groups $\cL^{\boldsymbol{\cdot}}(M,p)$ could be topologized as quotients in terms of a suitable function space topology on the set  $\Omega^{\text{ps}}(M,p)$. Our results suggest that a topological characterization of thin homotopy is also possible, with retracings as ``basic building blocks" through a limiting process, in the spirit of Whitney's approximation theorems \cite{Whi34}. We hope to return to these questions in the future.

In section \ref{sec:survey} we have tried to compile a sufficiently extensive list of references reflecting the historical development of the holonomy approach to gauge theory in geometry and physics. In this regard, a comprehensive compendium of the physics literature on the subject can be found in the book \cite{GP96}. A rather unfortunate feature of the existing mathematical literature on the subject is its evident scarcity and sparsity. We have attempted to collect what we consider to be a minimal comprehensive set of mathematical bibliography that illustrates our point, to the best of our abilities, while avoiding surveying the standard mathematical approach that originated from Atiyah's school \cite{Atiyah88}.\footnote{Besides Atiyah's many excellent expositions on the subject, another interesting mathema\-tical survey on the geometry of vector bundles and connections as a very convenient language in physics, and its relation to the theory of moduli spaces, is given in \cite{Var03}.} Given that the existing mathematical literature on the holonomy approach to gauge theory is the result of the work of many researchers and communities over the last 70 years, it is plausible that we could have made several important omissions. We apologize if that were the case. 

\vspace{3mm}

\noindent \textbf{Notes added on proof.} After the first version of this manuscript was made public, I was notified by Tamer Tlas of his work \cite{Tlas16}. Among other things, he proves a stronger version of theorem \ref{theo:smooth-thin}. This article has been modified to include his findings. After this article was in final form, I learned of Driver's work \cite{Dri89} which precedes most of the proofs of theorem \ref{theo:reconstruction}, and seems to have been unnoticed by their respective authors.

\section{Brief survey}\label{sec:survey}

The notions of holonomy and parallel transport in geometry and physics have a long and rich history that flourished through the entire twentieth century. From a physical perspective, the first occurrence of nonintegrable phase factors in the context of field theories with gauge symmetry can be traced back to the seminal work of Dirac \cite{Dir31} on the quantization of the electric charge. The fundamental question of the physical significance of electromagnetic potentials was first addressed from both an experimental \cite{AB59} and theoretical \cite{Man62a} viewpoints. Inspired by the pioneering ideas of Yang and Mills \cite{YM54} that led to the study of non-abelian gauge theories in physics, Bia\l{}ynicki-Birula \cite{BB63}  generalized Mandelstam's work to the case of non-abelian gauge fields, and determined an analog for the Yang--Mills equations in this context. Yet another pioneering idea was proposed by Wilson \cite{Wil74}, who introduced the notion of \emph{Wilson loops} as a fundamental complete set of gauge field observables in his study of the problem of confinement. It is rather significant that at the same time that the celebrated Wu--Yang dictionary \cite{WY75} was formulated\footnote{The appearance of the Wu--Yang dictionary was a catalyst for the birth of mathematical gauge theory, identifying the physical notion of a (classical) gauge field with the geometric notion of an orbit of connections on a principal bundle, under the action of the gauge group of bundle automorphisms. However, the first mathematician who studied the relation between gauge fields and connections on principal bundles was Robert Hermann \cite{Her70, Her75, Her78}. His insights on the subject have been largely ignored due to the unorthodox nature of his academic career.} ---  leading to the standard mathematical axiomatization of gauge field theories --- Yang \cite{Yang74} also postulated an ``integral fomalism" for non-abelian gauge fields based on the idea of non-abelian phase factors, which could be considered as the first general formulation of gauge theories from the holonomy perspective (in this respect, cf. \cite{Jack78}, \cite{Dol80}, \cite{MM81}, \cite{DL82}, \cite{ChT86,ChST86a,ChST86b} for similar considerations), as well as the first attempt at the reconstruction of gauge fields in terms of their non-abelian phase factors  (see \cite{Gil81} for a more detailed proposal --- at the physicists' level of rigor --- of the same idea formulated in terms of Wilson loop observables), while leading to non-abelian generalizations of Stokes' theorem \cite{Are80}, \cite{CCR99}. In yet another interesting direction, Aref'eva and Polyakov \cite{Are79,Pol80} observed that when an ordinary gauge field $A$ is interpreted as a map $H:\Omega^{\boldsymbol{\cdot}}(M,p)\rightarrow G$ in terms of its holonomy (a so-called \emph{chiral field on loop space}) the Yang--Mills equations for $A$ correspond to the zero curvature equation for the ``potential" $(\delta H) H^{-1}$.

\emph{Lattice gauge theory} (LGT) is a physical formalism for quantum gauge theories which originated in the work of Wilson \cite{Wil74}. Its basic idea is to replace a gauge field by a collection of ``parallel transporters" (elements in $G$) over a discrete set of paths in the base manifold (usually associated to a lattice in $\TT^{n}$), which play the role of a discretization of the former.  Although LGT is sometimes described as a pragmatic mechanism to model theoretical predictions (at the quantum level) that can be tested experimentally \cite{Creu83}, the insights it offers motivate the very appealing possibility of constructing quantum gauge theories by means of rigorously defined gauge field discretization mechanisms and continuum limits \cite{OS78}. 

All of the previously described ideas lead to many interesting mathematical questions. The differential geometry of connections on principal bundles \cite{KN63}, has ever since been considered as the standard approach towards a solid mathematical foundation of the subject. In this respect, the study of connections on principal bundles in terms of their holonomy seemed to have been motivated by topological questions in the study of fibre bundle theory \cite{Tel60,Tel63}, \cite{Milnor56}, \cite{Lash56, CL58}, \cite{Stash74,Stash12}, \cite{CM77}, while the first purely geometric aspects were considered by Kobayashi \cite{Koba54}, who posed the question of reconstruction of a connection on a principal bundle in terms of its holonomies for the first time.\footnote{This should be contrasted with the work of Narasimhan and Ramanan \cite{NR61}, who looked for an alternative reconstruction scheme in terms of a universal connection on the universal bundle of a Lie group $G$, which was analogous in spirit to the topological trends of that time.} Yet another fundamental mathematical precedent in the study of the differential geometry of loop spaces was initiated in the work of Chen \cite{Chen73,Chen77} where the idea of \emph{iterated integrals} is introduced.\footnote{The literature on Chen's iterated integrals is incredibly vast as the subject encountered many applications in both physics and mathematics; we only refer to \cite{Tav94}, \cite{HL10}, and \cite{Sta11} for subsequent developments that are directly relevant for the purposes of this work.}

Regarding topological aspects in the holonomy formulation of gauge theories, the work of L\"uscher \cite{Luscher82} initiated a series of investigations on the reconstruction of topological charges for a principal $G$-bundle purely in terms of LGT data \cite{Phil85,PS86,PS90,PS93}, which relies on a generic interpolation mechanism in the space of LGT data. A different account for the understanding of these problems is presented in \cite{MZ17,MZ19}, \label{sec:thin-not-ret} which addresses directly the characterization of the topology of principal $G$-bundles within the holonomy approach.

The holonomy formulation of gauge theory in physics experienced another boost in the 1980s due to several attempts to establish a theory of quantum gravity within the framework of representation theory for groups of based loops.\footnote{In this respect, the survey works \cite{Loll92,Loll94a,Loll94b} present a short summary of motivations, results and goals of the discipline, while the book \cite{GP96} is a rather comprehensive account of the developments of the subject in those decades.} From a more mathematical perspective, and building up on ideas of Gambini and Trias \cite{GT80,GT81,GT83,GT86,FGT85,GLT89,Gam91}, the school of Ashtekar initiated a programme that led to the development of the subject within the context of $C^{*}$-algebras \cite{AI92,AL93}, \emph{generalized connections} \cite{AL94,AL95}, and \emph{generalized measures} on the spaces $\cA_{P}/\cG_{P}$ of gauge orbits of smooth connections of a principal $G$-bundle $P$ \cite{Baez94,AMM94,MM95,BS97,Tlas14}. Generalized connections require to consider a suitable completion of $\cA_{P}/\cG_{P}$ (a rigorous mathematical treatment of generalized connections has also been given in \cite{Flei00,Flei03}). 
The study of the natural topological and differential geometric aspects of group structures on spaces of based loops on a smooth manifold, and their implications in the study of gauge field theories through holonomy, led to the introduction of the so-called \emph{extended loop groups} \cite{DiBGG93,DiBGGP94} and their generalized loop representations (cf. \cite{Schi96}, and the technical problems that this approach seemed to have unearthed). 

All of the previous developments provided many insights that paved the way towards a rigorous mathematical formulation of the reconstruction theorem 
of gauge fields in terms of their holonomy homomorphisms from a group of based loops on a manifold to a Lie group $G$. Different proofs of the reconstruction theorem appeared in \cite{Bar89,Bar91},  \cite{Lew93}, \cite{Haj93}, \cite{CP94}, leading to subsequent mathematical developments of the notion of groups of based loops \cite{Fulp94}, \cite{Gib97}, \cite{Wood97}, \cite{MP02}, \cite{SW09, CLW16}, \cite{Tlas16}. Somewhat independent, but equally important, are the works of Gross \cite{Gro85}, who gave an analytic proof of the equivalence between the Yang--Mills and Mandelstam--Bia\l{}ynicki-Birula equations, and Morrison \cite{Mor91}, who gave a characterization of connections on principal bundles over oriented Riemannian surfaces satisfying the Yang--Mills equations in terms of their corresponding holonomies. 
Regarding the differential geometry of loop spaces, more recently Stacey \cite{Sta09, Sta17} has studied the problem of defining manifold structures on spaces of piecewise-smooth based loops on a manifold, in the sense of the seminal work of Kriegl and Michor \cite{KM97}.\footnote{It is very tempting to inquire if the latter ideas and analytic techniques could be implemented to build a rigorous ``differential calculus" for smooth holonomy homomorphisms, that could serve as an alternative to the standard geometric analysis approach to mathematical gauge theory modeled on the quotient spaces $\cA_{P}/\cG_{P}$ (cf. \cite{Tav94}).}

\section{Thin homotopy and groups of based loops on a manifold}\label{sec:thin}

Let $M$ be a smooth connected manifold, on which a point $p\in M$ is fixed once and for all. A \emph{piecewise-smooth loop based at $p$} is a continuous map $\gamma:[0,1]\rightarrow M$ such that $\gamma(0)=\gamma(1)=p$, and for some finite subset $0=t_{0}<\dots<t_{r}=1$, $\gamma|_{[t_{i-1},t_{i}]}$ is smooth $\forall\; i=1,\dots,r$. Piecewise-analytic based loops are defined analogously.

Let $\Omega^{0}(M,p)$ denote the set of continuous loops in $M$ based at $p$. $\Omega^{0}(M,p)$ admits several natural function space topologies, the most common of which is the compact-open topology. Under such an intrinsic topological space structure, several important subspaces of $\Omega^{0}(M,p)$ arise if we impose additional smoothness constraints, i.e., the subspaces $\Omega^{\text{ps}}(M,p)\supsetneq\Omega^{\text{pa}}(M,p)$ of piecewise-smooth and piecewise-analytic loops, the subspaces $\Omega^{\infty}(M,p)\supsetneq\Omega^{\omega}(M,p)$ of smooth and analytic loops, etc.

$\Omega^{\boldsymbol{\cdot}}(M,p)$ will denote in general any subspace of $\Omega^{0}(M,p)$ as in the examples above. A homotopy between two elements $\gamma_{0},\gamma_{1}\in\Omega^{\boldsymbol{\cdot}}(M,p)$ will then be understood as a continuous map $\eta:[0,1]^{2}\rightarrow M$ such that $\eta(s,\cdot)\in\Omega^{\boldsymbol{\cdot}}(M,p)$ $\forall\; s\in[0,1]$, $\eta(i,\cdot)= \gamma_{i}$ for $i=0,1$, and that also satisfies the defining property of $\Omega^{\boldsymbol{\cdot}}(M,p)$, i.e., is smooth, piecewise-smooth on a polygonal paving of $[0,1]^{2}$, etc. By a slight abuse of notation, we will denote the trivial loop at $p$ simply by $p$.

The standard concatenation of an ordered pair of based loops $(\gamma,\gamma')$ in $\Omega^{0}(M,p)$, defined as
\[
\gamma'\cdot\gamma(t)= \left\{
\begin{array}{cl}
\gamma(2t) & \quad 0\leq t\leq 1/2,\\\\
\gamma'(2t-1) & \quad 1/2\leq t\leq 1,
\end{array}
\right.
\]
determines a binary operation on $\Omega^{0}(M,p)$, which together with the operation of loop inversion $\gamma^{-1}(t):=\gamma(1-t)$, descends into a group structure in the space of equivalence classes $\pi_{1}(M,p) := \Omega^{0}(M,p)/\sim_{H}$ under homotopy of based loops, giving rise to the fundamental group of $M$.  However, the fundamental group is just one example of a group structure induced on a space of equivalence classes in $\Omega^{0}(M,p)$ under a suitable homotopy relation, which is also the coarsest in a sense that will be made precise in definition \ref{def:group-like}. 
For this purpose, we will call a subspace $\Omega^{\boldsymbol{\cdot}}(M,p)\subset \Omega^{0}(M,p)$ \emph{concatenable} if concatenation of loops also defines a binary operation on it.  $\Omega^{\text{ps}}(M,p)$ and $\Omega^{\text{pa}}(M,p)$ are examples of concatenable subspaces, while $\Omega^{\infty}(M,p)$ and $\Omega^{\omega}(M,p)$ are not. This way, it is also possible to induce a group structure on different orbit spaces in a concatenable subspace with respect to weaker homotopy equivalence relations $\sim_{R}$.

\begin{definition}\label{def:group-like}
A \emph{group-like homotopy equivalence relation} in a concatenable subspace $\Omega^{\boldsymbol{\cdot}}(M,p)\subset \Omega^{0}(M,p)$ is an equivalence relation $\sim_{R}$ turning loop concatenation and inversion into a group epimorphism
\[
\Omega^{\boldsymbol{\cdot}}(M,p)/\sim_{R}\rightarrow \pi_{1}(M,p)\rightarrow e.
\]
\end{definition}

Several group-like homotopy equivalence relations are relevant to us, and are in fact implicit in the standard construction of the group structure on $\pi_{1}(M,p)$. The most relevant one for the purposes of this work --- thin homotopy --- was proposed by Barrett over the space $\Omega^{0}(M,p)$, although as done by Caetano--Picken, it suffices to define it over $\Omega^{\text{ps}}(M,p)$. 

By a \emph{generalized reparametrization} of a based loop $\gamma\in\Omega^{\text{ps}}(M,p)$, we mean a based loop of the form $\gamma\circ \phi$, where $\phi:[0,1]\rightarrow [0,1]$ is piecewise-smooth and such that $\phi(0)=0$, $\phi(1)=1$, and $0<\phi(t)<1$ if $0<t<1$. An analogous definition is given over $\gamma\in\Omega^{\text{pa}}(M,p)$. Any such a map leads to a degree 1 map $[\phi]:S^{1}\rightarrow S^{1}$ under the identification $S^{1}\cong\RR/\ZZ$, but in general the condition $\phi'(t)\neq 0$ is not imposed. By imposing suitable restrictions on $\phi$, we can think of generalized reparametrizations as maps $\cF_{\phi}:\Omega^{\boldsymbol{\cdot}}(\Omega,p)\rightarrow \Omega^{\boldsymbol{\cdot}}(\Omega,p)$. In particular, it follows that every piecewise-smooth based loop can be reparametrized into a smooth loop.

\begin{definition}[\cite{Bar91,CP94}]\label{def:thin-h}
Two based loops $\gamma_{0},\gamma_{1}\in \Omega^{0}(M,p)$ (resp.  $\Omega^{\text{ps}}(M,p)$ or $\Omega^{\text{pa}}(M,p)$) are called \emph{thin homotopic} if there is a homotopy $\eta$ between them (resp. piecewise-smooth or piecewise-analytic homotopy on a paving of $[0,1]^{2}$) such that 
\[
\eta\left([0,1]^{2}\right)\subseteq \gamma_{0}([0,1])\cup\gamma_{1}([0,1])\subset M.
\]
$\gamma_{0},\gamma_{1}$ are called \emph{thin equivalent} (denoted $\gamma_{0}\sim_{R_{\text{thin}}}\gamma_{1}$) if they are related by a finite composition of thin homotopies. A loop is \emph{thin} if it is thin equivalent to the trivial loop. 
\end{definition}

Hereafter we will mostly consider thin homotopies on $ \Omega^{\text{ps}}(M,p)$, unless stated otherwise. The simplest example of a thin homotopy on a based loop $\gamma\in\Omega^{\text{ps}}(M,p)$ is determined by any piecewise-smooth generalized reparametrization $\phi$, given by
\[
\eta(s,t)=\gamma(s\phi(t)+(1-s)t)
\]
Composition of generalized reparametrizations for the different classes of based loops leads to suitable equivalence relations in the corresponding spaces $\Omega^{\boldsymbol{\cdot}}(M,p)$, which we will call \emph{reparametrization equivalences} (denoted by $\sim_{R_{\text{rep}}}$). Even though reparametrization equivalence turns based loop concatenation into an associative structure on the set of equivalence classes $\Omega^{\text{ps}}(M,p)/\sim_{R_{\text{rep}}}$, it is still not an example of a group-like homotopy equivalence relation in $\Omega^{\text{ps}}(M,p)$.

Yet another example of thin homotopy, which suffices to define a group-like homotopy relation in $\Omega^{\text{ps}}(M,p)$, is that of a \emph{retracing}, which is a homotopy between a based loop of the form $\gamma'=\alpha_{2}\cdot\beta^{-1}\cdot\beta\cdot\alpha_{1}$ and the reduced form $\gamma=\alpha_{2}\cdot\alpha_{1}$, for any piecewise-smooth paths $\alpha_{1}$,$\alpha_{2}$, and $\beta$ in $M$ such that $\alpha_{1}(0)=\alpha_{2}(1)=p$ and $\alpha_{1}(1)= \alpha_{2}(0)=\beta(0)$.\footnote{The notion of loop retracing is due to Kobayashi \cite{Koba54}, who used the terms \emph{allongement} and \emph{accourcissement} for the operations transforming $\gamma$ into $\gamma'$ and viceversa (cf. \cite{Tel60}).} Such a homotopy is defined in terms of the standard homotopy between the loop $\beta^{-1}\cdot\beta$ and the constant loop at $\beta(0)$ given by $\eta(s,\cdot) =\beta^{-1}_{s}\cdot \beta_{s}$, where $\beta_{s}(t)=\beta((1-s)t)$.\footnote{Path concatenation is an associative operation only up to piecewise-smooth reparametrization. Hence the operation of retracing would be an unambiguously defined equivalence relation only over the space of equivalence classes in $\Omega^{\text{ps}}(M,p)$ under generalized reparametrization.} 

\begin{definition}[\cite{Koba54,Tel60,Lew93,Gib97,Spa01}]\label{def:gen-h}
Two based loops  $\gamma_{0},\gamma_{1}\in \Omega^{\text{ps}}(M,p)$ are called \emph{retrace equivalent} (denoted $\gamma_{0}\sim_{R_{\text{ret}}}\gamma_{1}$) if they differ by a finite number of retracings under implicit generalized reparametrization.\footnote{The weakening of reparametrizations is not necessary to define retrace equivalence. The justification for this weakening is discussed in remarks \ref{rem:ret->G} and \ref{rem:thin->rank-1}.} A based loop is \emph{retraceable} if it is retrace equivalent to the trivial loop.
\end{definition}

The \emph{holonomy} of a smooth connection $A$ \cite{KN63} on a smooth principal $G$-bundle $\pi:P\rightarrow M$ is a geometric notion associated to any piecewise-smooth loop in $M$ and a choice of base point $b\in\pi^{-1}(p)$. For any $A\in\cA_{P}$ we have an induced map
\[
\mathrm{Hol}_{A}:\Omega^{\text{ps}}(M,p)\rightarrow G,\qquad \gamma\mapsto \mathrm{Hol}_{A}(\gamma).
\]

\begin{definition}[\cite{Ash90,Lew93,Spa01}]\label{def:hoop-h}
Given a Lie group $G$, two based loops  $\gamma_{0},\gamma_{1}\in \Omega^{\text{ps}}(M,p)$ are \emph{$G$-equivalent} (denoted $\gamma_{0}\sim_{R_{G}}\gamma_{1}$) if $\mathrm{Hol}_{A}(\gamma_{0})= \mathrm{Hol}_{A}(\gamma_{1})$ for any smooth connection $A\in\cA_{P}$ on any smooth principal $G$-bundle $P\rightarrow M$. 
\end{definition}

\begin{remark}\label{rem:ret->G}
It is a standard fact that the holonomy of any smooth connection is an invariant under generalized reparametrizations and retracings of any given piecewise-smooth based loop. Hence we have the implication
\[
\mbox{Retrace equivalent} \quad\Rightarrow\quad \mbox{$G$-equivalent}
\]
This is one of several motivations for considering generalized reparametrizations $\phi:[0,1]\rightarrow[0,1]$ that may not be homeomorphisms (see remark \ref{rem:thin->rank-1}). 
Even though in general $G$-equivalence also depends on the Lie group $G$ by definition, in the case when $G$ is a connected and non-solvable matrix group retrace equivalence leads to strictly smaller classes (see sections \ref{sec:thin-not-ret} and \ref{sec:Tlas}). 
\end{remark}

Caetano--Picken defined yet another concatenable space of interest. A smooth based loop $\gamma$ is said to have \emph{sitting instants at $p$} if there is $0<\varepsilon \leq 1/2$ such that $\gamma([0,\varepsilon]\cup[1-\varepsilon,1])=p$. We will denote by $\Omega^{\infty}_{0}(M,p)\subset\Omega^{\infty}(M,p)$ the subspace of smooth loops with sitting instants at $p$.

\begin{definition}[\cite{CP94}]\label{def:intim-h}
Two based loops $\gamma_{0},\gamma_{1}\in\Omega^{\infty}_{0}(M,p)$ are \emph{intimate} (denoted $\gamma_{0}\sim_{R_{\text{int}}}\gamma_{1}$) if there is a smooth homotopy $\eta:[0,1]^{2}\rightarrow M$ such that 
\begin{itemize}
\item[(i)] $\mathrm{rk}\left(d\eta_{(s,t)}\right)\leq 1 \quad \forall\; (s,t)\in[0,1]^{2}$ (rank-1 homotopy condition),
\item[(ii)] there is $0<\varepsilon<1/2$ such that $\forall\; s\in[0,1]$, $\eta(s,\cdot)\in\Omega^{\infty}_{0}(M,p)$ with respect to $\varepsilon$, and moreover $\eta(s,\cdot)=\gamma_{0}$ if $s\in[0,\varepsilon]$, and $\eta(s,\cdot)=\gamma_{1}$ if $s\in[1-\varepsilon,1]$.
\end{itemize} 
\end{definition}

\begin{proposition}[\cite{CP94}]\label{prop:rep-smooth}
For any given $\gamma\in\Omega^{\text{\emph{ps}}}(M,p)$ there exists a generalized reparametrization $\phi:[0,1]\rightarrow[0,1]$ such that $\gamma\circ\phi \in \Omega^{\infty}_{0}(M,p)$. 
\end{proposition}

\begin{remark}\label{rem:thin->rank-1}
Every thin homotopy $\eta$ for any given pair $\gamma_{0},\gamma_{1}\in\Omega^{\text{ps}}(M,p)$ is a particular case of a rank-1 homotopy \cite{CP94}.  It follows that for every thin homotopy $\eta$, there exists a smooth map $\psi:[0,1]^{2}\rightarrow [0,1]^{2}$ such that $\psi(s,\cdot )$ is a generalized reparametrization $\forall\;s\in[0,1]$, $\eta\circ\psi(s,\cdot )\in \Omega^{\infty}_{0}(M,p)$ $\forall\; s\in[0,1]$, and condition (ii) in definition \ref{def:intim-h} is satisfied, i.e., every thin homotopy in $\Omega^{\text{ps}}(M,p)$ can be deformed into an intimate homotopy in $\Omega^{\infty}_{0}(M,p)$ through a family of generalized reparametrizations. In particular, it follows that every thin loop in $\Omega^{\infty}_{0}(M,p)$ is intimate to the trivial loop. Thus, the equivalence relation $\sim_{R_{\text{int}}}$ can be thought of as a ``taming" of thin equivalence in $\Omega^{\text{ps}}(M,p)$ in terms of reparametrization equivalence. Intimacy can also be defined on spaces of based maps $\gamma:[0,1]^{k}\to M$, leading to generalizations of higher homotopy groups \cite{CP98, MP10}.
\end{remark}

The proof of the following result is relatively straightforward to verify. Its proof is also justified in each of the original articles that introduced each corresponding equivalence relation.

\begin{proposition}\label{prop:group-like}
The equivalence relations introduced in definitions \ref{def:thin-h}, \ref{def:gen-h}, \ref{def:hoop-h} and \ref{def:intim-h} are examples of group-like homotopy equivalence relations on their corresponding concatenable subspaces.
\end{proposition}

Consequently, there is an induced topological group structure on the quotients
\begin{align*}
& \cL^{\text{thin}}(M,p):=\Omega^{\text{ps}}(M,p)/\sim_{R_{\text{thin}}},\qquad & \cL^{\text{ret}}(M,p):=\Omega^{\text{ps}}(M,p)/\sim_{R_{\text{ret}}},\\\\
& \cL^{G}(M,p):=\Omega^{\text{ps}}(M,p)/\sim_{R_{G}},\qquad &\cL^{\text{int}}(M,p):=\Omega^{\infty}_{0}(M,p)/\sim_{R_{\text{int}}}.
\end{align*}
When no distinction is necessary, we will refer to any of these groups as $\cL^{\boldsymbol{\cdot}}(M,p)$.\footnote{We will refer to these topological groups generically as \emph{groups of based loops} in $M$. In the physics literature, the group $\cL^{G}(M,p)$ is commonly referred to as the \emph{group of hoops} for the pair $(M,G)$.}

Generalized reparametrizations and retracings of piecewise-smooth based loops are particular examples of thin homotopies in $\Omega^{\text{ps}}(M,p)$. It follows from remark \ref{rem:thin->rank-1} that we have two induced group epimorphisms
\[
\Phi:\cL^{\text{ret}}(M,p)\rightarrow \cL^{\text{thin}}(M,p),\qquad \Psi: \cL^{\text{thin}}(M,p)\rightarrow \cL^{\text{int}}(M,p)
\]
Taking several inclusions and projections into account, these group epimorphisms can be arranged into the following commutative diagram 

\[
\xymatrixrowsep{.6cm}
\xymatrixcolsep{0cm}
\xymatrix{
\Omega^{\infty}_{0}(M,p) \ar@{}[ddd]^(.17){}="a"^(.83){}="b" \ar "a";"b"\ar@{}[rd]^(.25){}="a"^(.75){}="b" \ar "a";"b" \ar@{}[rrrr]^(.15){}="a"^(.85){}="b" \ar "a";"b" & & & & \Omega^{\text{ps}}(M,p)\ar@{}[dl]^(.25){}="a"^(.75){}="b" \ar "a";"b"\ar@{}[ddd]^(.17){}="a"^(.83){}="b" \ar "a";"b" \\
&\Omega^{\infty}_{0}(M,p)/\sim_{R_{\text{rep}}}\ar@{}[ddl]^(.25){}="a"^(.75){}="b" \ar "a";"b"&\cong&\Omega^{\text{ps}}(M,p)/\sim_{R_{\text{rep}}}\ar@{}[d]^(.25){}="a"^(.75){}="b" \ar "a";"b" &\\
&&&\cL^{\text{ret}}(M,p) \ar@{}[dr]^(.3){}="a"^(.7){}="b" \ar "a";"b"&\\
\cL^{\text{int}}(M,p) & & & & \cL^{\text{thin}}(M,p) \ar@{}[llll]^(.15){}="a"^(.85){}="b" \ar "a";"b"
}
\]

\begin{definition}
A \emph{$G$-holonomy relation} in $\Omega^{\boldsymbol{\cdot}}(M,p)\subset\Omega^{\text{ps}}(M,p)$ is any equiva\-lence relation $\sim_{R}$ for which the holonomy of any smooth connection on any smooth principal $G$-bundle is an invariant on equivalence classes. 
\end{definition}

\begin{proposition}[\cite{Koba54, CP94}]
Retrace equivalence in $\Omega^{\emph{ps}}(M,p)$ (definition \ref{def:gen-h}) and intimacy in $\Omega^{\infty}_{0}(M,p)$ (definition \ref{def:intim-h}) are examples of $G$-holonomy relations for any Lie group $G$.
\end{proposition} 

The motivation of the notion of a $G$-holonomy relation is the following. For any such an equivalence relation, the holonomy map of any smooth connection $A\in\cA_{P}$ would descend into a group homomorphism
\[
\mathrm{Hol}_{A}: \Omega^{\boldsymbol{\cdot}}(M,p)/\sim_{R}\;\rightarrow G.
\]
The prototypical example of a $G$-holonomy relation is $G$-equivalence itself. One of the main objectives of the works \cite{Bar91,Lew93, CP94} is to show that each corresponding group-like homotopy equivalence relation is a $G$-holonomy relation, in such a way that the following reconstruction theorem holds (cf. \cite{Haj93}).

\begin{theorem}[Reconstruction of gauge fields as holonomy homomorphisms]\label{theo:reconstruction}
For any group homomorphism $H:\cL\rightarrow G$ satisfying a suitable smoothness condition on smooth families of based loops, where $\cL= \cL^{\text{\emph{thin}}}(M,p)$ \cite{Bar91}, $\cL^{\text{\emph{ret}}}(M,p)$ \cite{Lew93}, or $\cL^{\text{\emph{int}}}(M,p)$ \cite{CP94} there is a smooth principal $G$-bundle $\pi:P\rightarrow M$ and a smooth connection $A\in\cA_{P}$ such that $\mathrm{Hol}_{A}= H$ in terms of a choice $b\in\pi^{-1}(p)$.
\end{theorem}

\begin{remark}
As pointed out by Caetano--Picken in \cite{CP94}, the definition of thin homotopy in \cite{Bar91} makes the realization of thin equivalence as a holonomy relation a nontrivial problem. The introduction of the spaces $\Omega^{\infty}_{0}(M,p)$ and the intimacy equivalence relation on them may be thought of as a sort of `gauge-fixing' condition for smooth based loops, that helps simplifying the proof of theorem \ref{theo:reconstruction}. It is implicit from the existence of an epimorphism $\Psi:\cL^{\text{thin}}(M,p)\rightarrow \cL^{\text{int}}(M,p)$ for any smooth manifold $M$, that thin equivalence of piecewise-smooth based loops should be a $G$-holonomy relation for any Lie group $G$. In this respect, the results of Tlas (\cite{Tlas16}; section \ref{sec:Tlas}) provide a strengthening of this idea since in particular they imply that $\Psi$ is actually an isomorphism.

\end{remark}

\subsection{A smooth thin loop that is not retraceable}\label{sec:thin-not-ret}

Thin homotopy and retrace homotopy are not equivalent when defined over piecewise-smooth loops, although the latter is always a particular case of the former. This is illustrated by a minimal example of a thin smooth loop $\gamma$ in $\RR^{2}$ based at the origin that is not retraceable. The example can be constructed explicitly in terms of the bump function $f:[0,1]\rightarrow [0,\infty)$ defined as
\begin{equation}\label{eq:bump}
f(t)=\left\{
\begin{array}{ll}
e^{-1/t(1-t)} & \quad \text{if}\quad t\in (0,1),\\\\
0 & \quad\text{if}\quad t=0,1.
\end{array}
\right.
\end{equation}
whose one-sided $n$th-derivatives at 0 and 1 are identically zero for any $n\in\NN$. Given any pair of sequences $\{r_{n}\}\subset (0,\infty)$, $\{\theta_{n}\} \subset[0,2\pi)$ such that 
\[
\lim\limits_{n\rightarrow \infty} r_{n}\rightarrow 0,\qquad \theta_{n}<\theta_{n+1}\quad\forall \; n\in\NN,
\]
we define the smooth loop $\gamma\in\Omega^{\infty}\left(\RR^{2},0\right)$ in a piecewise manner as
\[
\gamma|_{\left[1-2^{-n+1},1-2^{-n}\right]}:= f\circ\phi^{-1}_{n} \cdot \left(r_{n}\cos\left(\theta_{n}\right),r_{n}\sin\left(\theta_{n}\right)\right),
\]
for $n=1,2,\dots$, where
\[
\phi_{n}:[0,1]\rightarrow \left[1-2^{-n+1},1-2^{-n}\right],\qquad \phi_{n}(t):= \left(1-2^{-n}\right)t + \left(1-2^{-n+1}\right)(1-t).
 \] 
By construction, we have that $\gamma^{-1}(0) = \{1-1/2^{n} \;:\; n\in\NN\}\cup\{0,1\}$, and a thin homotopy between $\gamma$ and the trivial loop is simply given by $\eta(s,\cdot):=(1-s)\gamma$, but clearly $\gamma$ is not a reparametrization of a finite concatenation of piecewise-smooth retraceable loops  (figure \ref{fig:thin}). Notice that $\gamma$ would have bounded variation under the standard metric in $\RR^{2}$ if and only if the series $\sum_{n=1}^{\infty}r_{n}$ is convergent.

\begin{figure}[!ht]
\caption{The image of a thin smooth loop $\gamma\in \Omega^{\infty}\left(\RR^{2},0\right)$}\label{fig:thin}

\vspace*{7mm}

\centering
\includegraphics[width=0.9\textwidth]{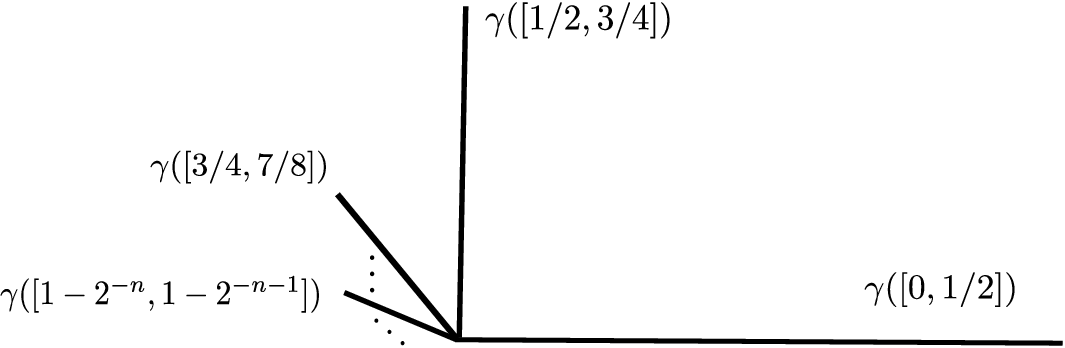}

\vspace*{3mm}
\end{figure}

\subsection{Results of Tlas}\label{sec:Tlas}

In \cite{Tlas16}, Tlas considered the space $\Omega^{1}_{0}(M,p)$ of $C^{1}$-based loops in $M$ with vanishing derivative at end points, and studied the homotopy equivalence relation on it  induced from compositions of rank-1 homotopies. A crucial ingredient is a factorization theorem (theorem 1 in \cite{Tlas16}; cf. theorem \ref{theo:smooth-thin}) for any based $C^{1}$-loop, which in particular associates a transfinite word to it, and leads to the notion of a \emph{whisker}, which is any element in $\Omega^{1}_{0}(M,p)$ whose reduced word is trivial. Similarly, the following result is proved:

\begin{theorem}[\cite{Tlas16}, theorem 3]\label{theo:Tlas} 
The following conditions on a loop $\gamma\in\Omega^{1}_{0}(M,p)$ are equivalent:
\begin{itemize}
\item[(i)] $\gamma$ is a whisker,
\item[(ii)] $\gamma$ is equivalent to the trivial loop via a rank-1 homotopy,\footnote{In \cite{Tlas16}, rank-1 homotopies are called \emph{thin homotopies}. This does not represent a notational inconsistency in view of remark \ref{rem:equiv-thin-int}.}
\item[(iii)] For any semi-simple Lie group $G$, the holonomy of $\gamma$ for every smooth connection on every principal $G$-bundle $P$ is trivial.\footnote{In particular, this indicates that theorem 5 in \cite{Spa01}, which states the isomorphism $$\cL^{\text{ret}}(M,p)\cong\cL^{G}(M,p)$$ when $G$ is connected and non-solvable, is not valid. This is the case since every non-retraceable thin loop (as in section \ref{sec:thin-not-ret}) will define a nontrivial element in the kernel of the epimorphism $\Phi: \cL^{\text{ret}}(M,p)\rightarrow\cL^{\text{thin}}(M,p)$ following from remark \ref{rem:thin->rank-1} and proposition \ref{prop:group-like}, and consequently $\cL^{\text{ret}}(M,p)\not\cong\cL^{\text{thin}}(M,p)$. Notice that Spallanzani's work suggests that Tlas' theorem is also valid for any connected and non-solvable Lie group $G$.}
\end{itemize}
\end{theorem}

\begin{remark}\label{rem:equiv-thin-int}
Thin homotopy is a special case of a rank-1 homotopy \cite{CP94}, and after suitable generalized reparametrizations, it follows from Tlas' factorization theorem that every whisker in $\Omega^{\text{ps}}(M,p)$ is a thin loop. Therefore, it also follows from theorem \ref{theo:Tlas} that two loops in $\Omega^{\text{ps}}(M,p)$ are  thin homotopic if and only if they are homotopic through a rank-1 homotopy, and we have the isomorphism
\begin{equation}
\cL^{\text{thin}}(M,p)\cong \cL^{\text{int}}(M,p).
\end{equation}
Moreover, when $G$ is semi-simple, thin homotopy is the largest $G$-holonomy relation possible in $\Omega^{\text{ps}}(M,p)$, i.e. $\cL^{\text{thin}}(M,p)= \cL^{G}(M,p)$. 
\end{remark}

\section{Results}\label{sec:results}

We will provide several structural results on the comparison between thin equivalence and retrace equivalence of piecewise-smooth based loops. Ultimately, Tlas' theorem (theorem \ref{theo:Tlas}) implies a stronger version of our main result (theorem \ref{theo:smooth-thin}). Nevertheless, we believe that our approach is still interesting, since working with thin homotopies rather that rank-1 homotopies provides a natural mechanism for generalizations to the space $\Omega^{0}(M,p)$ (remark \ref{rem:cont}), while the simplicity of the definition of thin homotopy makes many considerations more elementary.

Our strategy consists of studying the structure of the sets $\gamma^{-1}(0)$. This gives a simple proof of the Ashtekar--Lewandowski theorem \cite{AL93} on the equivalence of thin and retrace homotopies for piecewise-analytic loops. First, we provide a general topological result concerning thin homotopy on continuous based loops.

\begin{lemma}\label{lemma:thin-int}
If $\gamma_{0},\gamma_{1}\in\Omega^{0}(M,p)$ are thin equivalent, then $\gamma_{0}([0,1])\cap\gamma_{1}([0,1])$ is the image of a based loop $\gamma\in\Omega^{0}(M,p)$ which is thin equivalent to $\gamma_{0}$ and $\gamma_{1}$.
\end{lemma}
\begin{proof}
The images $\gamma_{0}([0,1])$ and $\gamma_{1}([0,1])$ are closed since $[0,1]$ is compact and $M$ is Hausdorff. Hence $I_{i}:=\gamma^{-1}_{i}\left(\gamma_{0}([0,1])\cap\gamma_{1}([0,1])\right)$ is closed in $[0,1]$ for $i=0,1$. The complements $[0,1]\setminus I_{i}$ are open, and hence equal to an at most countable disjoint union of open intervals in $[0,1]$. Consider the case $i=0$ and any such a subinterval $(a,b)\subset [0,1]\setminus I_{0}$. Then we have that $\gamma_{0}((a,b))\subset \gamma_{0}([0,1])\setminus\left(\gamma_{0}([0,1])\cap\gamma_{1}([0,1])\right)$, while $\gamma_{0}(a),\gamma_{0}(b)\in\gamma_{0}([0,1])\cap \gamma_{1}([0,1])$. Since $\gamma_{0}$ and $\gamma_{1}$ are thin homotopic, this is only possible if $\gamma_{0}(a)=\gamma_{0}(b)=p_{0}$. Any thin homotopy would then collapse the image $\gamma_{0}([a,b])$ onto $p_{0}$. It follows that the restriction $\gamma_{0}|_{[a,b]}$ is necessarily a thin loop based at $p_{0}$. The map
\[
\gamma(t) =\left\{
\begin{array}{cl}
\gamma_{0}(t) & \text{if}\quad t\in I_{0},\\\\
\gamma_{0}(a)=\gamma_{0}(b) & \text{if}\quad t\in(a,b)\subset [0,1]\setminus I_{0}
\end{array}
\right.
\]
is continuous by construction, and such that $\gamma([0,1])= \gamma_{0}([0,1])\cap\gamma_{1}([0,1])$. It follows that $\gamma_{0}$ is thin homotopic to $\gamma$, and consequently $\gamma_{1}$ is too.
\end{proof}

Given a based loop $\gamma\in \Omega^{\text{ps}}(M,p)$, we will associate the following closed subsets of $[0,1]$ to it: 
\[
I_{\gamma}:=\gamma^{-1}(p),\qquad\qquad
J_{\gamma}:=\{t\in[0,1]\;:\; \gamma\; \text{is smooth at $t$ and\;\;} d\gamma|_{t} =0\}.
\]
As we have seen in section \ref{sec:thin-not-ret}, it is easy to provide examples of smooth based loops $\gamma$ where $I_{\gamma}$ and $J_{\gamma}$ are not finite subsets. General closed subsets of $[0,1]$ could be rather complicated (e.g. Cantor sets). This complexity is not excluded from the subsets $I_{\gamma}$ and $J_{\gamma}$.  Namely, we have the following result.

\begin{proposition}\label{prop:I-J}
For any pair of closed subsets $I,J\subset [0,1]$, there exist smooth loops $\gamma,\gamma'\in\Omega^{\infty}(M,p)$ such that (i) $I_{\gamma} = I$;  (ii) $I_{\gamma'}=\{0,1\}$ and $J_{\gamma'}=J$. 
\end{proposition}

\begin{proof}
It is sufficient to verify the claim when $M = \RR^{2}$ and $p=0$. The proof is modeled on a celebrated theorem of H. Whitney, which states that for any closed subset $I\in\RR$, there is a smooth function $f:\RR\rightarrow \RR$ such that $f^{-1}(0)=I$. The basic idea for the proof of Whitney's theorem is to express the open set $[0,1]\setminus I$ as a countable disjoint union of open intervals 
\[
[0,1]\setminus I_{\gamma} = \bigsqcup_{n\in\cI} (a_{n},b_{n}).
\]

Over the closure $[a_{n},b_{n}]$ of any such an interval, we define the restriction of the loop $\gamma$ to be 
\[
\gamma|_{[a_{n},b_{n}]}:= f\circ\phi^{-1}_{n} \cdot \left(r_{n}\cos\left(\theta_{n}\right),r_{n}\sin\left(\theta_{n}\right)\right),
\]
where $\phi_{n}(t):= b_{n}t + a_{n}(1-t)$, $r_{n}\in(0,1]$, $\theta_{n}\in (0,\pi/2)$ are arbitrary, and $f$ is the bump function \eqref{eq:bump}. Then (i) readily follows. As for (ii), the same mechanism yields a smooth loop $\gamma:[0,1]\rightarrow \RR^{2}$ such that $I_{\gamma}=\{0,1\}$ and $\{t\in[0,1]\;:\; d\gamma|_{t}=0\}=J$ by consideration of a primitive (whose existence would depend on imposing further constraints on $\{r_{n}\}$) and scaling it by the bump function $f$.
\end{proof}

\begin{remark}
Even though each restriction $\gamma|_{[a_{n},b_{n}]}$ for any connected component $(a_{n},b_{n})\subset [0,1]\setminus I_{\gamma}$ of any piecewise smooth loop $\gamma$ leads to a piecewise-smooth loop $\gamma'$ for which $I_{\gamma'}=\{0,1\}$, it follows from proposition \ref{prop:I-J} that the subset $J_{\gamma'}$ could still in principle be an arbitrary closed subset in $[0,1]$. Therefore, based smooth loops on a manifold $M$ could in general traverse the base point $p$ in an arbitrarily wild way, and even over any connected component $(a,b)\subset[0,1]\setminus I_{\gamma}$, the set $J_{\gamma}\cap[a,b]$ could be arbitrarily wild as well. 
\end{remark}

\begin{lemma}\label{lemma:nowhere-dense}
Every based loop $\gamma\in\Omega^{\text{\emph{ps}}}(M,p)$ is thin homotopic to a loop $\gamma'$ for which $I_{\gamma'}$ is nowhere dense. 
\end{lemma}

\begin{proof}
Every connected component in $I_{\gamma}$ corresponds to a closed subinterval (which could be a point) where the image of $\gamma$ is $p$. Consider the collection of all closed subintervals $I\subset I_{\gamma}$ which are not a point. We can construct a 1-parameter family of reparametrizations that shrink any such an $I$ into a point. The resulting based loop $\gamma'$ would be thin homotopic to $\gamma$ by construction, such that $I_{\gamma'}$ is nowhere dense, and there would be a induced bijection $[0,1]\setminus I_{\gamma}\leftrightarrow [0,1]\setminus I_{\gamma'}$.
\end{proof}

\begin{lemma}\label{lemma:thin-I-gamma}
If a based loop $\gamma\in\Omega^{\text{\emph{ps}}}(M,p)$ is thin, then there exists a thin homotopy $\eta$ between $\gamma$ and the trivial loop such that $I_{\gamma^{s}}=I_{\gamma}$ $\forall \;s\in[0,1)$, where $\gamma^{s}:=\eta(s,\cdot)$. Moreover, $\eta$ could be chosen to be piecewise-analytic if $\gamma$ is piecewise-analytic.
\end{lemma}

\begin{proof}
Since $\gamma$ is a thin based loop, there exists a continuous and piecewise-smooth map $\eta:[0,1]^{2}\rightarrow M$ such that $\gamma^{s}\in\Omega^{\text{ps}}(M,p)$, $\gamma_{0}=\gamma$, $\gamma_{1}=p$, and $\eta\left([0,1]^{2}\right)=\gamma([0,1])$. Then for each fixed $t_{0}\in[0,1]$, $\delta_{t_{0}}:=\eta(\cdot,t_{0})$ is a path within $\gamma([0,1])$ connecting $\gamma(t_{0})$ and $p$. Assume that $t_{0}\notin I_{\gamma}$. Then $\delta_{t_{0}}(0)\neq p$. Define
\[
s_{t_{0}}:=\min \{\;s\in[0,1]:\;\delta_{t_{0}}(s)=p\}. 
\] 
Define a new thin homotopy $\eta':[0,1]^{2}\rightarrow M$ by first considering restrictions and reparametrizations along each $[0,s_{t_{0}}]$ for each $t_{0}\in[0,1]\setminus I_{\gamma}$. By the piecewise-smoothness of $\eta$, the extension of $\eta'$ to each $t_{0}\in I_{\gamma}$ given by $\eta'(\cdot,t_{0})=p$ would be continuous and piecewise-smooth, and piecewise-analytic when $\eta$ is so. By construction $\eta'$ is a thin homotopy between $\gamma$ and the trivial loop such that $I_{\gamma'^{s}}=I_{\gamma}$ $\forall\; s\in[0,1)$, and which is piecewise-smooth, or piecewise-analytic when $\eta$ is so. 
\end{proof}

The complexity of a thin based loop $\gamma\in\Omega^{\text{ps}}(M,p)$ is encoded in the sets $I_{\gamma}$ and $J_{\gamma}$. The general structure of $I_{\gamma}$ and $J_{\gamma}$ is considerably simpler along the subspace $\Omega^{\text{pa}}(M,p)$ of piecewise-analytic based loops, as illustrated in the subsequent theorem and its corollary (cf. \cite{AL93}). These results illustrate the fundamental difference between analyticity and smoothness of based loops with respect to thin homotopy.

\begin{theorem}\label{theo:thin-analytic}
A loop $\gamma\in \Omega^{\text{\emph{pa}}}(M,p)$ is thin if and only if it is reparametrization equivalent to a finite concatenation of retraceable loops $\gamma_{1},\dots,\gamma_{m}\in\Omega^{\text{\emph{pa}}}(M,p)$.
\end{theorem}

\begin{proof}
It suffices to verify that if $\gamma$ is thin, then there exists a piecewise-analytic reparametrization $\gamma'=\gamma\circ\phi$ such that $\gamma'$ is a finite concatenation of retraceable piecewise-analytic loops. Hence we can assume without any loss of generality that $\gamma$ is analytic. If $\gamma$ is the trivial loop, then there is nothing to prove. Otherwise, the analyticity of $\gamma$ implies that $I_{\gamma}$ is finite, i.e., $I_{\gamma}=\{t_{0}=0 < t_{1}<\dots < t_{m}=1\}$. It then follows that $\gamma = \gamma_{m}\cdot\dots\cdot \gamma_{1}$ is a concatenation of $m$ analytic loops $\gamma_{1},\dots,\gamma_{m}$, such that $I_{\gamma_{i}} = \{0,1\}$ for $i=1,\dots,m$. Each $\gamma_{i}$ is explicitly given as 
\[
\gamma_{i} = \gamma|_{[t_{i-1},t_{i}]}\circ \phi_{i},\qquad \phi_{i}(t) :=  t_{i-1}(1-t) + t_{i} t, \qquad i=1,\dots,m.
\]
We claim that each $\gamma_{i}$ is reparametrization equivalent to a based loop in $\Omega^{\text{pa}}(M,p)$ of the form $\beta^{-1}_{i}\cdot\beta_{i}$, for some piecewise-analytic analytic path $\beta_{i}$ such that $\beta_{i}(0)=p$. 

Consider a piecewise-analytic thin homotopy $\eta$ between $\eta(0,\cdot)=\gamma$ and the trivial loop $\eta(1,\cdot)=p$ as in lemma \ref{lemma:thin-I-gamma}. Then $\eta(s,t_{i})=p$ for each $i=1,\dots,m-1$, and $\eta(s,t)\neq p$ if $s\in(0,1)$ and $t\in [0,1]\setminus I_{\gamma}$. We would also have that $\gamma^{s} = \gamma^{s}_{m}\cdot\dots\cdot \gamma^{s}_{1}$ $\forall\; s\in[0,1]$. It follows that each $\gamma_{i}$, $i=1,\dots,m$, is also thin and analytic, and the sets $J_{\gamma_{i}}$ are finite. Since for each $i=1,\dots, m$, we have that $\gamma^{s}_{i}\left([0,1]\right)\subset \gamma_{i}\left([0,1]\right)$ $\forall\; s\in[0,1]$, then each $J_{\gamma_{i}}\setminus\{0,1\}$ is necessarily nonempty. Define
\begin{align*}
t_{\gamma_{i}}:=&\min\{t\in J_{\gamma_{i}}\setminus\{0,1\}\;:\;\gamma_{i}\left(\left[0,t\right]\right)=\gamma_{i}([0,1])\}\\\\
=&\min\{t\in (0,1)\;:\;\gamma_{i}\left(\left[0,t\right]\right)=\gamma_{i}([0,1])\},
\end{align*}
and let
\[
\beta'_{i}(t) := \gamma_{i}\left(t_{\gamma_{i}}t\right).
\]
Then there exists a generalized analytic reparametrization $\varphi_{i}:[0,1]\rightarrow [0,1]$ such that $\beta'_{i}=\beta_{i}\circ \varphi_{i}$ with $d\beta_{i}|_{t}\neq 0$ $\forall\; t\in(0,1)$.
If we express $\gamma_{i}=\delta_{i}\cdot(\beta_{i}\circ\varphi_{i})$, it follows from the definition of thin homotopy that $\delta_{i}$ and $\beta_{i}$ have the same image in $M$. Similarly, it follows that there is a generalized analytic reparametrization $\psi_{i}:[0,1]\rightarrow [0,1]$ such that $\delta_{i} = \beta^{-1}_{i}\circ\psi_{i}$. Hence we can express $\gamma_{i}= \left(\beta^{-1}_{i}\cdot \beta_{i}\right)\circ \vartheta_{i}$ for a piecewise-analytic reparametrization $\vartheta_{i}$. 
\end{proof}

\begin{corollary}
Two based loops $\gamma_{0},\gamma_{1}\in\Omega^{\text{\emph{pa}}}(M,p)$ are thin equivalent if and only if they are retrace equivalent (cf. \cite{AL94}). 
\end{corollary}

\begin{proof}
It remains to show that thin equivalence implies retrace equivalence. In general, $\gamma_{0}([0,1])\cap\gamma_{1}([0,1])$ is a closed subset of $M$. It follows from the piecewise-analyticity of $\gamma_{0}$ and $\gamma_{1}$ that the closed subsets $I_{i}:=\gamma_{i}^{-1}\left(\gamma_{0}([0,1])\cap\gamma_{1}([0,1])\right)$, $i=0,1$, are a finite union of disjoint closed intervals in $[0,1]$, i.e.,
\[
I_{i} = \bigsqcup_{j=0}^{r_{i}} \left[c^{i}_{j},d^{i}_{j}\right],\quad  i=0,1,
\]
where $c^{i}_{j}\leq d^{i}_{j}$ for each $j=0,\dots,r_{i}$.
Clearly $c^{i}_{0}=0$  and $d^{i}_{r_{i}}=1$. If $\gamma_{0}$ and $\gamma_{1}$ are thin homotopic, it follows from the general arguments in the proof of lemma \ref{lemma:thin-int} that
\[
\gamma_{i}\left(I_{i}\right)= \gamma_{i}([0,1]),\quad i=0,1,
\]
and $\gamma_{i}\left(d^{i}_{j-1}\right) = \gamma_{i}\left(c^{i}_{j}\right)=p^{i}_{j}$ for each $i=0,1$, $j=1,\dots, r_{i}$. Hence each restriction
\[
\gamma_{i}|_{\left[d^{i}_{j-1},c^{i}_{j}\right]} 
\]
is a piecewise-analytic loop based at $p^{i}_{j}$, which is thin in the piecewise-analytic sense. By theorem \ref{theo:thin-analytic}, it is reparametrization equivalent to a finite concatenation of retraceable loops. Therefore, $\gamma_{0}$ and $\gamma_{1}$ are retrace equivalent.
\end{proof}

The examples of thin smooth based loops that are not retraceable that we have constructed before are clearly non-analytic.
In particular, it follows that there is no analog of theorem \ref{theo:thin-analytic} for piecewise-smooth based loops. The next result indicates that, nevertheless, the complexity of a piecewise-smooth thin loop $\gamma$ is entirely encoded in the complexity of the set $I_{\gamma}$. This is so since along the closure of every connected component of its complement $[0,1]\setminus I_{\gamma}$, the restriction of $\gamma$ would necessarily reduce to a reparametrization of an elementary retraceable loop.

\begin{theorem}\label{theo:smooth-thin}
If a piecewise-smooth based loop $\gamma\in \Omega^{\text{\emph{ps}}}(M,p)$ is thin, then for every connected component $(a,b)\subset [0,1]\setminus I_{\gamma}$, the restriction $\gamma|_{[a,b]}$ is reparametrization equivalent to a retraceable loop of the form $\beta^{-1}\cdot\beta$ for some smooth path $\beta:[0,1]\rightarrow M$ such that $\beta^{-1}(p)=\{0\}$.
\end{theorem}

\begin{proof}
We can assume that $\gamma$ is smooth after reparametrization. By hypo\-thesis $\gamma(a)=\gamma(b)=p$. Therefore $\gamma':=\gamma|_{[a,b]}\circ \phi$, where $\phi(t)= bt+a(1-t)$, is a smooth based loop, and moreover $I_{\gamma'}=\{0,1\}$. Consider any thin homotopy $\eta$ between $\gamma$ and the trivial loop such that $I_{\gamma^{s}}= I_{\gamma}$ $\forall s\in[0,1)$ as in lemma \ref{lemma:thin-I-gamma}. The induced restrictions $\gamma'^{s}:=\gamma^{s}|_{[a,b]}\circ \phi$ also determine a thin homotopy between $\gamma'$ and the trivial loop. Hence $\gamma'$ is a thin loop. Similarly to the proof of theorem \ref{theo:thin-analytic}, we consider
\[
t_{\gamma'}:=\min\{t\in(0,1)\;:\;\gamma'([0,t])=\gamma'([0,1])\}
\]
(notice that in the smooth case, $\min \{t\in J_{\gamma'}\setminus\{0,1\}\}$ may not exist at all). By continuity it follows that, moreover,
\[
\gamma'\left([t_{\gamma'},1]\right) = \gamma'\left([0,t_{\gamma'}]\right) = \gamma'([0,1]).
\]
The smoothness of $\gamma'$ implies that the set $\gamma'([0,t_{\gamma'}])$ could be parametrized as a path $\beta:[0,1]\rightarrow M$ such that $\beta(0)=0$ and $\beta(1)=\gamma'\left(t_{\gamma'}\right)$, for which $d\beta|_{t} \neq 0$ if $t\in(0,1)$. This could be done in terms of arc length rectification for an auxiliary Riemannian metric in $M$. It follows that the image of the retraceable loop $\beta^{-1}\cdot\beta$ coincides with the image of $\gamma'$. It moreover follows that there is a piecewise-smooth generalized reparametrization $\varphi:[0,1]\rightarrow[0,1]$ such that $\gamma' = (\beta^{-1}\cdot\beta)\circ\varphi$.
\end{proof}

\begin{remark}\label{rem:cont}
It can be verified that the proofs of lemmas \ref{lemma:nowhere-dense} and \ref{lemma:thin-I-gamma}, as well as theorem \ref{theo:smooth-thin} and corollary \ref{cor:non-eq}, can be also adapted to hold for thin homotopy defined over the larger spaces $\Omega^{0}(M,p)$.
\end{remark}

We conclude this section with the following corollaries on the structure of thin equivalence for piecewise-smooth based loops. The proof of corollary \ref{cor:non-eq} follows as an immediate consequence of proposition \ref{prop:I-J} and theorem \ref{theo:smooth-thin}. 

\begin{corollary}\label{cor:non-eq}
For any $\gamma\in\Omega^{\text{\emph{ps}}}(M,p)$, there is $\gamma'\in\Omega^{\text{\emph{ps}}}(M,p)$ such that $\gamma\sim_{\text{\emph{thin}}}\gamma'$ but $\gamma\nsim_{\text{\emph{ret}}}\gamma'$. Consequently, for any $\gamma\in\Omega^{\text{\emph{ps}}}(M,p)$, the equivalence class $[\gamma]_{\text{\emph{thin}}}$ is strictly larger than the equivalence class $[\gamma]_{\text{\emph{ret}}}$.\footnote{In contrast, for any $\gamma\in\Omega^{\text{pa}}(M,p)$, the equivalence classes of piecewise-analytic loops $[\gamma]_{\text{thin}}$ and $[\gamma]_{\text{ret}}$ coincide (cf. \cite{AL94}).} 
\end{corollary}

\begin{corollary}\label{cor:limit}
Every pair of thin equivalent loops $\gamma\sim_{R_{\text{\emph{thin}}}}\gamma'$ in $\Omega^{\text{\emph{ps}}}(M,p)$ is the limit of a sequence of pairs of retrace equivalent loops $\{\gamma_{n},\gamma'_{n}\}$ with respect to the compact-open topology, in the sense that $\{\gamma_{n}\}\rightarrow \gamma$ and $\{\gamma'_{n}\}\rightarrow \gamma'$ as $n\rightarrow \infty$, and $\gamma_{n}\sim_{R_{\text{\emph{ret}}}}\gamma'_{n}$ $\forall\; n\in\NN$.
\end{corollary}

\begin{proof}
As before, it is sufficient to consider the case when $\gamma\in \Omega^{\text{ps}}(M,p)$ is a thin loop. The general case of a pair $(\gamma,\gamma')$ of thin homotopic loops in $\Omega^{\text{ps}}(M,p)$ will follow analogously.
Consider any given thin loop $\gamma\in\Omega^{\text{ps}}(M,p)$, and let $\eta$ be a thin homotopy betwen $\gamma$ and the trivial loop such that $I_{\gamma^{s}}=I_{\gamma}$ $\forall\; s\in[0,1)$ as in lemma \ref{lemma:thin-I-gamma}. As before, let
\[
[0,1]\setminus I_{\gamma} = \bigsqcup_{n\in\cI} (a_{n},b_{n}).
\]
If $\cI$ is a finite set then there is nothing to prove. Hence we will assume that $\cI=\NN$. For any given $n\in\NN$, define $\gamma^{s}_{n}:[0,1]\rightarrow M$ as
\[
\gamma^{s}_{n}(t) =\left\{
\begin{array}{cc}
\gamma^{s}(t)& \quad \text{if}\quad \displaystyle t\in\bigsqcup_{k\leq n} (a_{k},b_{k}),\\\\
p & \quad\text{otherwise}.
\end{array}
\right.
\]
and let $\gamma_{n}:=\gamma^{0}_{n}$. It readily follows that $\gamma^{s}_{n}\in\Omega^{\text{ps}}(M,p)$ $\forall\; n\in\NN$ and $s\in[0,1]$, and each $\gamma_{n}$ is retraceable in terms of the sequence of thin homotopies $\eta_{n}$ defined by the families $\gamma^{s}_{n}$. By construction, $\gamma^{s}_{n}(I_{n})=\{p\}$, where $I_{n}= [0,1]\setminus \bigsqcup_{k\leq n} (a_{k},b_{k})$, while $\gamma^{s}_{n}\vert_{\bigsqcup_{k\leq n}[a_{k},b_{k}]} = \gamma^{s}\vert_{\bigsqcup_{k\leq n}[a_{k},b_{k}]}$  $\forall\; s\in[0,1]$.

A basis element $\cB(I,\cU)$ of the compact-open topology in $\Omega^{\text{ps}}(M,p)$ is determined by a closed set $I\subset[0,1]$ and an open set $\cU\subset M$, and consists of all loops $\gamma$ such that $\gamma(I)\subset \cU$. Assume that $\gamma\in\cB(I,\cU)$ for some $I$ and $\cU$. Since
\[
I = \left(I\cap I_{n}\right) \cup \bigsqcup_{k\leq n}\left(I \cap  (a_{k},b_{k})\right),
\]
there are two possibilities. When $I\cap I_{\gamma}\neq \emptyset$, it follows that $\gamma^{s}_{n}(I)\subset \cU$ $\forall\; n\in \NN$, i.e.  $\gamma^{s}_{n}\in\cB(I,\cU)$ $\forall\; n\in\NN,\; s\in[0,1]$. Similarly, if $I\cap I_{\gamma} = \emptyset$, then $I\subset \bigsqcup_{k\leq n_{0}} (a_{k},b_{k})$ for some $n_{0}\in\NN$. Therefore $I\cap I_{n} = \emptyset$ $\forall\; n> n_{0}$ and $\gamma^{s}_{n}\in\cB(I,\cU)$ $\forall\; n> n_{0},\; s\in[0,1]$. In both cases $\{\gamma_{n}\}\rightarrow \gamma$ in the compact-open topology. 
\end{proof}

\begin{remark}\label{rem:BV}
We can turn $M$ into a metric space by means of an auxiliary Riemannian metric on it. Under any such a choice, the compact-open topology on $\Omega^{\text{ps}}(M,p)$ is metrizable in terms of
\[
d(\gamma,\gamma'):= \sup\{d(\gamma(t),\gamma'(t))\;:\; t\in[0,1]\}.
\]
Let $\Omega^{\text{ps}}_{b}(M,p)\subset \Omega^{\text{ps}}(M,p)$ denote the subspace of loops with bounded variation, which is directly dependent on the choice of Riemannian metric in $M$. These subspaces play a fundamental role in the construction of smooth structures on spaces of piecewise-smooth based loops \cite{Sta17}. 
\end{remark}

\begin{corollary}\label{cor:non-Hausdorff}
The quotient topology on the groups $\cL^{\text{\emph{ret}}}(M,p)$ induced from the compact-open topology in $\Omega^{\text{\emph{ps}}}(M,p)$ is never Hausdorff.
\end{corollary}

\begin{proof}
 it follows from corollary \ref{cor:limit} that for every $\gamma\in\Omega^{\text{ps}}(M,p)$, $[\gamma]_{\text{ret}}\subsetneq [\gamma]_{\text{thin}} \subset\overline{[\gamma]_{\text{ret}}}$ and consequently $[\gamma]_{\text{ret}}\neq \overline{[\gamma]_{\text{ret}}}$ for any $\gamma\in\Omega^{\text{ps}}(M,p)$. In parti\-cular, the subspace $\Omega^{\text{ps}}(M,p)\sim_{\text{ret}} \Omega^{\text{ps}}(M,p)\subset \Omega^{\text{ps}}(M,p)\times \Omega^{\text{ps}}(M,p)$ of retrace equivalent pairs in $\Omega^{\text{ps}}(M,p)$ is never closed. 
\end{proof}

\begin{remark}
An alternative proof of the equivalence of (i) and (iii) in theorem \ref{theo:Tlas}, concerning the correspondence between thin equivalence and $G$-equivalence for a semi-simple Lie group $G$ in $\Omega^{\text{ps}}(M,p)$, is also hinted by corollary \ref{cor:limit}. 
Since the holonomy of any smooth connection on a smooth principal $G$-bundle $P\rightarrow M$, together with a choice of fiber point $b\in\pi^{-1}(p)$, defines a continuous map $\Omega^{\text{ps}}(M,p)\rightarrow G$ satisfying Barrett's smoothness conditions \cite{Bar91}, it would be sufficient to show that $\overline{[\gamma]_{\text{thin}}} = [\gamma]_{\text{thin}}$ for any given $\gamma\in\Omega^{\text{ps}}(M,p)$. 
\end{remark}

\section{Conclusion and further remarks}

One of the motivating purposes of this work was to shed light on the intricacies of thin equivalence in spaces of piecewise smooth loops on a given smooth manifold. Since any satisfactory and mathematically rigorous approach to the study of gauge theories from the holonomy perspective would necessarily rely on functional analytic tools developed on loop spaces, a fundamental step into that direction would be to study these intricacies in careful detail. In this concluding section, we have compiled a short list of remarks and natural questions which arise when following this paradigm. The common feature that unifies all of these questions is the ultimate need to consider natural topologies on the spaces $\Omega^{\text{ps}}(M,p)$ and their induced quotients.

As illustrated by corollary \ref{cor:non-Hausdorff}, the induced quotient topologies on groups of based loops will depend drastically on the nature of the different equivalence classes of a given based loop $\gamma\in\Omega^{\text{ps}}(M,p)$. It is natural to ask if the subspace of thin equivalence pairs in $\Omega^{\text{ps}}(M,p)\times \Omega^{\text{ps}}(M,p)$ is closed under a suitable choice of function space topology in $\Omega^{\text{ps}}(M,p)$. Otherwise, this would indicate that $\Omega^{\text{ps}}(M,p)$ is ``too big", and a suitable subspace of it should be found.

\begin{question}
Are the subspaces of thin equivalent pairs in $\Omega^{\text{\emph{ps}}}(M,p)$ closed in $\Omega^{\text{\emph{ps}}}(M,p)\times \Omega^{\text{\emph{ps}}}(M,p)$ under the compact-open topology in $\Omega^{\text{\emph{ps}}}(M,p)$? If not, what is the largest concatenable subspace $\Omega^{\boldsymbol{\cdot}}(M,p)\subset \Omega^{\text{\emph{ps}}}(M,p)$ leading to a Hausdorff quotient under a suitable function space topology?
\end{question}

\begin{question}
Is there a Hausdorff subgroup $\cL^{\boldsymbol{\cdot}}(M,p)\subset \cL^{\text{\emph{thin}}}(M,p)$ (under a suitable function space topology) for which the reconstruction theorem \ref{theo:reconstruction} holds? If so, what is the smallest subgroup with these properties?
\end{question}

Once a Hausdorff topology has been fixed on the group $\cL^{\boldsymbol{\cdot}}(M,p)$, we can consider the induced groups of continuous topological group homomorphisms to $G$, which we will denote succinctly by $\Hom(\cL^{\boldsymbol{\cdot}}(M,p),G)$. These groups contain a subgroup $\Hom(\cL^{\boldsymbol{\cdot}}(M,p),G)_{0}$ consisting of those homomorphisms homotopic to the trivial homomorphism. Then we have that
\[
\Hom(\cL^{\boldsymbol{\cdot}}(M,p),G)/\Hom(\cL^{\boldsymbol{\cdot}}(M,p),G)_{0}\cong \pi_{0}\left(\Hom(\cL^{\boldsymbol{\cdot}}(M,p),G)\right).
\]
Let $\check{H}^{1}(M,G)$ denote the space of equivalence classes of principal $G$-bundles over $M$. Perhaps more interestingly, it is also very reasonable to expect to have a conjectural bijection (cf. \cite{MZ17, MZ19})
\begin{equation}\label{eq:bij-bun}
\Hom(\cL^{\boldsymbol{\cdot}}(M,p),G)/\Hom(\cL^{\boldsymbol{\cdot}}(M,p),G)_{0}\cong \check{H}^{1}(M,G).
\end{equation}

\begin{question}
Under the previous assumptions on $\Omega^{\boldsymbol{\cdot}}(M,p)$, do we always have a bijection \eqref{eq:bij-bun}?
\end{question}

Let $\Hom^{\infty}(\cL^{\boldsymbol{\cdot}}(M,p),G)$ denote the subspace of $\Hom(\cL^{\boldsymbol{\cdot}}(M,p),G)$ satisfying Barrett's smoothness conditions (H3) \cite{Bar91,CP94}. The reconstruction of gauge fields in terms of holonomy homomorphisms  for a given pair $(M,G)$, described in theorem \ref{theo:reconstruction}, can be expressed succinctly as a statement of the bijectivity of the induced map
\begin{equation}\label{eq:hol-corr}
\bigsqcup_{\{P\}\in\check{H}^{1}(M,G)}\cA_{P}/\cG_{P}\rightarrow \Hom^{\infty}(\cL^{\boldsymbol{\cdot}}(M,p),G)/G,
\end{equation}
where $G$ acts on $\Hom^{\infty}(\cL^{\boldsymbol{\cdot}}(M,p),G)$ by conjugation, corresponding to changes of base point $b\in \pi^{-1}(p)\subset P$ under right $G$-action.
While the left-hand side is not an infinite dimensional manifold, it is a stratified space containing an open smooth (infinite dimensional) part  \cite{MV81}, and in particular, it carries a natural topology. This comparison leads to the following natural and fundamental question.

\begin{question} 
Under the previous assumptions on $\Omega^{\boldsymbol{\cdot}}(M,p)$, is the bijection \eqref{eq:hol-corr} a homeomorphism?
\end{question}

Following the same train of thought, it would also be desirable to reinterpret Barrett's smoothness conditions (H3) in terms of a suitable ``smooth structure" on the induced loop group $\cL^{\boldsymbol{\cdot}}(M,p)$ that is compatible with that of the quotients $\cA_{P}/\cG_{P}$, providing a topological refinement of the correspondence \eqref{eq:hol-corr}. It would be interesting to explore this possibility within the context of the work of Stacey \cite{Sta09, Sta17}, who introduced a refining condition on spaces of piecewise-smooth loops in such a way that they admit an infinite dimensional manifold structure in the sense of Kriegl-Michor \cite{KM97}.

\vspace{3mm}

\noindent \textbf{Acknowledgments.} The author would like to thank Jos\'e A. Zapata for many discussions that led to the creation of this work, Roger F. Picken for sharing his opinion and insights on the subject, Tamer Tlas for providing remarks leading to a revision of this work, and the referee for the careful review and the very pertinent suggestions made. The author was supported by the DFG SPP 2026 priority programme ``Geometry at infinity".

\bibliographystyle{plain}
\bibliography{Loops}

\end{document}